\documentclass[10pt,conference,compsocconf]{IEEEtran}

\IEEEoverridecommandlockouts

\usepackage{url,amsmath,graphicx,amssymb,stfloats,subfigure,fancybox}
\usepackage{amsfonts}
\usepackage{algorithmic}
\usepackage{enumerate}
\usepackage[mathscr]{eucal}
\usepackage{verbatim,pifont,color,amssymb,amsfonts,amsmath,graphicx,pgf,slashbox,psfrag,ifpdf}
\usepackage[english]{babel}
\usepackage{epsfig}
\usepackage{pause}
\usepackage{epsf}
\usepackage{multicol}
\usepackage{xcolor}
\usepackage{amsthm,multirow}

\newtheorem{theorem}{Theorem}
\newtheorem{lemma}{Lemma}

\newtheorem{definition}{Definition}

\begin{document}

\title{Optimal Deadline Scheduling with Commitment}

\author{Shiyao Chen$^\dag$\thanks{$^\dag$School of ECE, Cornell
University, Ithaca NY 14853.} \quad Lang Tong$^\dag$
\quad Ting He$^\ddag$
\thanks{$^\ddag$IBM T.J. Watson Research Center,
Hawthorne, NY 10532. \hfill \break This research was sponsored in part by the
U.S. National Institute of Standards and Technology under Agreement Number 60NANB10D003
and the National Science Foundation under a Grant CNS-1135844.}
}

\maketitle

\begin{abstract}
We consider an online preemptive scheduling problem where jobs with deadlines arrive sporadically. A commitment
requirement is imposed such that the scheduler has to either accept or decline a job immediately upon arrival. The scheduler's decision to accept an arriving job constitutes a contract with the customer; if the accepted job is not completed by its deadline as promised, the scheduler loses the value of the
corresponding job and has to pay an additional penalty depending on
the amount of unfinished workload. The objective of the online scheduler is to maximize the overall profit, \textit{i.e.}, the total value of the admitted jobs completed before their deadlines less the penalty paid for
the admitted jobs that miss their deadlines.  We show that the maximum competitive ratio is $3-2\sqrt{2}$
and propose a simple online algorithm to achieve this competitive ratio.  The optimal scheduling includes
a threshold admission and a greedy scheduling policies.  The proposed algorithm has direct applications to
the charging of plug-in hybrid electrical vehicles (PHEV) at garages or parking lots.
\end{abstract}


\section{Introduction}

In a conventional setting of deadline scheduling, jobs arrive sporadically, each with prescribed processing time, deadline, and value. Upon arrival, jobs are queued until their respective deadlines, during which time an online scheduler can schedule any pending jobs in the system. In general, there is no guarantee that a submitted job will be completed by its deadline. In fact, a customer who submits a job does not know whether the job will be completed until after the deadline. For example, an online scheduler may accept a job into the system but later choose to work on another more profitable job instead.

In this paper, we consider a variation of the deadline scheduling problem by imposing a commitment requirement at the arrival of a job. In particular, if a job is accepted and successfully completed, the scheduler receives a certain reward. If the scheduler is unable to complete an accepted job, it pays a penalty. The scheduler receives neither reward nor penalty if it declines a job upon arrival.

The deadline scheduling problem considered in this paper has a direct application in scheduling the charging of plug-in hybrid electric vehicles (PHEV) in parking lots or garages. In this case,
each car arrives with a certain charge level, and the customer has some idea about how long the car can be left at the facility (\textit{e.g.}, approximately the duration in which the customer will be shopping before picking up the car).  Upon submitting a request, the customer is either turned away or told the car will be charged at a certain price. And the customer is informed that if the car is not charged to the requested level by the deadline, a compensation will be made (\textit{e.g.}, with a voucher for future charges).

While the requirement of immediate commitment reduces the number of unsatisfactory customers
and the amount of penalty, it brings nontrivial complications to the deadline scheduling problem. The difficulty comes from the fact that optimal decision on whether to turn away a customer seems to depend on the kind of jobs to arrive in the future. Had the scheduler known that there is a highly profitable job to arrive, it would have declined some of the less profitable ones. Our goal is to maximize the profit by optimally trading off accepting more customers against avoiding excessive penalties due to unfinished jobs. To accommodate general arrivals and workload, we aim at optimizing the \emph{competitive ratio} that characterizes the worst-case profit relative to that of an optimal offline scheduling algorithm, for which we establish the optimal competitive ratio and give an online scheduling algorithm to achieve the optimum competitive ratio.

\subsection{Related Work}

Without the commitment requirement, there is a considerable literature on the deadline scheduling problem, starting from the seminal work of Liu and Layland \cite{Liu&Layland:JACM73}. The problem is often divided into the underloaded and overloaded regimes. The former corresponds to the case when there exists an offline scheduling algorithm that can complete all jobs arrived whereas the latter corresponds to the case when some jobs cannot be completed even for the best offline scheduling algorithm. For the underloaded scenarios, it has been
shown that simple online scheduling algorithms such as earliest deadline first (EDF)
\cite{Liu&Layland:JACM73,Dertouzos:IFIP74} and least laxity first (LLF) \cite{Mok:83} achieve the same performance as the optimal
offline scheduling algorithm.
The assumption of underloaded overall workload, however, is restrictive
and unverifiable in practice. Locke showed in \cite{Locke:86} that
both EDF and LLF can perform poorly in the presence of overload.
There were efforts to develop an online scheduling algorithm with
performance guarantee in terms of competitive ratio (see Definition \ref{def:cr}), even
when the system is overloaded. Online scheduling algorithms with competitive ratio $1/4$
were proposed in \cite{Baruah:RTS92,Koren&Shasha:SIAMJC95} and $1/4$ was proved to be optimal competitive-ratio-wise for the deadline scheduling problem without commitment.

One of the first work that proposes the idea of commitment is \cite{BarNoy:03} (commitment is termed as immediate notification), in which Bar-Noy \textit{et. al.} considered the application of video on demand where customers submit movie request and the scheduler manages to either accept or decline the request within a specific ``notification time" after the request releases. Bar-Noy \textit{et. al.} studied the competitive ratio when the ``notification time" varies from zero (immediate notification) to proportional to the length of the movie requested.

Later, scheduling with immediate notification and immediate decision has been studied in \cite{Goldwasser:JS03} (single processor, immediate notification), \cite{Ding:AAIM06} (multiple processors, immediate notification), and \cite{Ding:ESA07,Ebenlendr:WAOA08} (multiple processors, immediate decision).
Immediate decision requires that, in addition to providing to the customers an immediate feedback regarding admission or declination, the scheduler also has to provide to the customer upon job release the specific scheduled time of the job, if accepted.
The proportional value model was considered in \cite{Goldwasser:JS03}, while \cite{Ding:ESA07,Ding:AAIM06} considered the unit length jobs with unit value.
An online scheduling algorithm with immediate decision is proposed in
\cite{Ding:ESA07} with asymptotic competitive ratio $(e-1)/e$, while the authors of \cite{Ebenlendr:WAOA08} showed $(e-1)/e$ to be an asymptotic upper bound of any online algorithms.
However, the authors of \cite{Goldwasser:JS03,Ding:AAIM06,Ding:ESA07,Ebenlendr:WAOA08} dealt with non-preemptive scheduling with no non-completion penalty involved.

With the commitment requirement online preemptive scheduling with
deadlines becomes much more challenging in the presence of overload.
The authors of \cite{Thibault&Laforest09} and the author of \cite{Fung:AICCC10} gave separately two preemptive scheduling algorithms for multiple processors with immediate notification and non-completion
penalty with the proportional
value model ($v_i=p_i$, see Section \ref{sec:formulation}). In \cite{Thibault&Laforest09,Fung:AICCC10} the non-completion
penalty associated with a job with value $v_i=p_i$ is set as $\rho v_i$ with the penalty parameter $\rho\geq0$.
The competitive ratio results given in \cite{Thibault&Laforest09,Fung:AICCC10}
are $(\min_{a>1+\rho}(2a+3)(1+\frac{\rho}{a-\rho}+\frac{1}{a-\rho-1}))^{-1}$ and $(2\rho+3+2\sqrt{\rho^2+3\rho+2})^{-1}$, respectively.
Even for the situation with no non-completion penalty $(\rho=0)$,
the competitive ratio results given in \cite{Thibault&Laforest09,Fung:AICCC10}
reduces to $(\min_{a>1}(2a+3)(1+\frac{1}{a-1}))^{-1}=(7+2\sqrt{10})^{-1}$ and $(3+2\sqrt{2})^{-1}$, respectively, which is at most as good as our result $3-2\sqrt{2}$ in this paper for single processor with non-completed portion penalized. On the other hand, there are no arguments in \cite{Thibault&Laforest09,Fung:AICCC10} establishing upper bounds of competitive ratio ever achievable to quantify how far the proposed algorithms are away from optimality.

There is a series of work by Hou, Borkar and Kumar \cite{Hou&Borkar&Kumar:INFOCOM09,Hou&Kumar:MobiHoc09,Hou&Kumar:INFOCOM10} and Jaramillo, Srikant and Ying \cite{Jaramillo&etal:JSAC11,Jaramillo&Srikant:ToN11} dealing with the deadline scheduling problem with a different setup from that adopted in this paper.
Specifically, the channel (the counterpart in \cite{Hou&Borkar&Kumar:INFOCOM09,Hou&Kumar:MobiHoc09,Hou&Kumar:INFOCOM10,Jaramillo&etal:JSAC11,Jaramillo&Srikant:ToN11} of the processor in this paper) is modeled as a stationary, irreducible Markov process with a finite state space (unreliable channel model), whereas the processor is always dedicated to scheduling the jobs in this paper.
Due to the unreliable channel model, the packet transmission (the counterpart of the job in this paper) may take a random amount of time to go through, whereas the job length in this paper is deterministic upon arrival. Each client (transmitter) specifies a delay requirement (transmissions which take longer than the delay requirement is invalid), which corresponds to the deadlines in this paper. The packet arrival process is assumed to be a stationary, irreducible Markov process with finite state space for each client, whereas the job arrival process can be arbitrary and quite bursty in this paper.
Thus, the stochastic model of the processor (channel), the job (packet) arrival process and the job length (packet transmission duration) is available in \cite{Hou&Borkar&Kumar:INFOCOM09,Hou&Kumar:MobiHoc09,Hou&Kumar:INFOCOM10,Jaramillo&etal:JSAC11,Jaramillo&Srikant:ToN11}. On the other hand, the job arrival as well as the job length can be arbitrary for the future job released in this paper. There is also difference in the metric used; the feasibility optimality is studied in \cite{Hou&Borkar&Kumar:INFOCOM09,Hou&Kumar:MobiHoc09,Hou&Kumar:INFOCOM10,Jaramillo&etal:JSAC11,Jaramillo&Srikant:ToN11}, \textit{i.e.}, the overall packet arrival is assumed to be underloaded, whereas the overloaded scenario is treated in this paper with the metric competitive ratio.

The problem of PHEV charging scheduling in public garage has been considered in \cite{Kulshrestha&etal:PES09,Tulpule&etal:11,Su&etal:PES11}. An energy economic analysis of PHEV charging using solar photovoltaic panels at workplace parking garage is conducted in \cite{Tulpule&etal:11} with the conclusion that PHEV charging facility in public garage is beneficial to both the car owners as well as the facility operator.
The authors of \cite{Kulshrestha&etal:PES09} aggregated a system architecture model, an operation model and a PHEV battery model to simulate PHEV charging in a municipal parking lot. The method of particle swarm optimization is employed to allocate energy to PHEVs in
\cite{Su&etal:PES11}. The performance of the scheduling policies proposed in \cite{Kulshrestha&etal:PES09,Tulpule&etal:11,Su&etal:PES11} are validated via simulation results. This paper adopts a deadline scheduling framework with non-completion penalty that suits well for PHEV charging application and proposes an online scheduling algorithm with worst case performance guarantee.

\subsection{Summary of Results}
In this paper, we impose a penalty on unfinished workload and obtain results on the optimal competitive ratio for the online preemptive deadline scheduling with commitment. We propose an online scheduling algorithm DSC (acronym for Deadline Scheduling with Commitment) with
competitive ratio $3-2\sqrt{2}=17.16\%$. We also provide a converse via an adversary argument and show that no online scheduling algorithm exists with a better competitive ratio, thus further establishing the optimality of DSC competitive-ratio-wise. Comparing with the optimal competitive ratio $1/4=25\%$ without the commitment requirement in \cite{Baruah:RTS92,Koren&Shasha:SIAMJC95}, we observe a performance loss of $7.84\%$ competitive-ratio-wise with the additional commitment obligation.

\section{Problem Formulation}\label{sec:formulation}

A job $T=(r,p,d,v)$ is represented by a quadruple specified by the release time $r$, processing time $p$, deadline $d$, and value $v$.  We assume the so-called proportional value model \cite{Goldwasser:JS03} where the value $v$ of a job is proportional to its processing time $p$, or without loss of generality, $v=p$. A job $T$ is called \textit{tight} if $r+p=d$, which implies that the scheduler must either work on the job or decline it immediately. Preemption is allowed at no cost in scheduling (\textit{i.e.}, a preempted job
can be resumed from the point of preemption at a later time).
In our scheduling problem an input instance $I$ includes $n$ jobs $T_1,\ldots,T_n$ to be scheduled on a single processor, where the integer
$n$ is the total number of jobs released for instance $I$ (the total number of jobs can differ over input instances).
In general, we are interested in a collection of instances in the input instance set $\mathcal{I}$.

Use $\mathcal{S}_{\mbox{\tiny online}}$ to denote online scheduler and $\mathcal{S}_{\mbox{\tiny offline}}$ the offline scheduler.
An online scheduler $\mathcal{S}_{\mbox{\tiny online}}$ knows the parameters of job $T_i$ only at the release time $r_i$. Deadlines are firm, \textit{i.e.}, completing a job after its deadline yields zero value.
The scheduling is
done with commitment, \textit{i.e.}, upon the release of each job,
the scheduler has to decide whether to accept or decline the job
request.
Each accepted job incurs a non-completion penalty equal to the unfinished workload (shortage) if it is not completed by its deadline.
The profit obtained by the
scheduler is the total value of all completed jobs, minus all penalties paid. This specific non-completion penalty suits the application of PHEV charging well since the utility is delivered to the car owner continuously as the battery charging level increases, unlike some computing jobs in high performance computing grids for which the utility can be obtained only upon the completion of all the computation steps.

Given an instance $I$, we denote by $S_{\mbox{\tiny online}}(I)$ ($S_{\mbox{\tiny offline}}(I)$) as the total profit (or value) obtained by the scheduler $\mathcal{S}_{\mbox{\tiny online}}$ ($\mathcal{S}_{\mbox{\tiny offline}}$).
Our objective is to make the online scheduler competitive across all instances in $\mathcal{I}$.

In contrast to the online scheduler, an offline scheduler $\mathcal{S}_{\mbox{\tiny offline}}$ is
clairvoyant and knows the entire input instance a priori. Due
to the prior knowledge of the job parameters, the offline scheduler
is able to make commitment decisions.
We denote by $\mathcal{S}^\ast_{\mbox{\tiny offline}}$ the optimal offline scheduler.

The problem is to design an online scheduling algorithm with worst-case performance
guarantee (relative to the optimal offline scheduling algorithm) even in the presence of overload. The performance
guarantee is given in terms of competitive ratio defined below.

\begin{definition}\label{def:cr}
\textbf{Competitive ratio:} An online algorithm $\mathcal{S}_{\mbox{\tiny online}}$ is
$\alpha$-competitive for an input instance set $\mathcal{I}$ if
$\min_{I\in\mathcal{I}}\frac{S_{\mbox{\tiny online}}(I)}{S^\ast_{\mbox{\tiny offline}}(I)}\geq\alpha$
where $I$ varies over all possible input instances in $\mathcal{I}$.
\end{definition}

That is, an $\alpha$-competitive online algorithm is guaranteed to
achieve at least $\alpha$ fraction of the optimal offline value
under any input instance $I$ in the input instance set $\mathcal{I}$. For the rest of the paper
the input instance set $\mathcal{I}$ is fixed to be the set of all input instances $I$ such that $I$ contains finite number of jobs and each job satisfies $d_i\geq r_i+p_i$ (otherwise, neither the online nor the offline scheduler is able to complete the job by its deadline and the job can thus be deleted from the input instance $I$). Note that both underloaded and overloaded input instances are included in the input instance set $\mathcal{I}$ defined above.

\section{Optimal Deadline Scheduling with Commitment}

Compared with the traditional deadline scheduling without the commitment requirement,
the additional difficulty imposed by the commitment obligation to the online scheduler depends on the overall workload of the jobs released: when the overall workload is underloaded, simple scheduling algorithms such as EDF and LLF achieve competitive ratio 1 by simply admitting all jobs released; the restriction to the underloaded case precludes the need of admission control.
In this section, we describe an online scheduling algorithm DSC dealing with the overloaded scenario and establish its optimality in competitive ratio. Specifically, we show $3-2\sqrt{2}$ as the optimal worst-case performance of online algorithms relative to the (optimal) offline counterpart. We summarize these results in the following theorem followed by a detailed description of DSC.
\begin{theorem}\label{thm:1}
For the input instance set $\mathcal{I}$ specified in Section \ref{sec:formulation},
\begin{enumerate}
\item The competitive ratio $3-2\sqrt{2}$ is achievable by DSC.
\item $3-2\sqrt{2}$ upper bounds the competitive ratio ever achievable by any online scheduling algorithms.
\end{enumerate}
\end{theorem}

In other words, there is a loss of $7.84\%$ competitive-ratio-wise with the additional commitment obligation when we
compare the optimal competitive ratio $1/4=25\%$ without the commitment requirement (\cite{Baruah:RTS92,Koren&Shasha:SIAMJC95}) with
the result in Theorem \ref{thm:1} ($3-2\sqrt{2}=17.16\%$).

\subsection{DSC Scheduling Algorithm}\label{sec:algo}

The key idea behind DSC is to evaluate the admission decision based on the comparison of the potential profit associated with accepting and declining a job, if the job is ``difficult" to accommodate into the current schedule.
Even assuming the scheduler accepts the job just released, there are plenty of alternatives in the specific schedule of the job just released as well as the other pending jobs in the system (due to the acceptance of the new job, it may be necessary to update the schedule of the other jobs).
DSC evaluates the profits associated with the two options by restricting to one alternative in the many ways of updating of the schedule after accepting the newly released job. Specifically, if the newly released job can be appended in the end of the current schedule while still being within its deadline (the job is ``easy" to accommodate into the current schedule, see the blue, green and red jobs in Fig. \ref{fig:0}, Fig. \ref{fig:1} and Fig. \ref{fig:52}, respectively for illustration), the job is accepted and appended in the end of the current schedule. Otherwise (the job is ``difficult" to accommodate into the current schedule, see the green and brown jobs in Fig. \ref{fig:3} and Fig. \ref{fig:6}, respectively for illustration), the two options are weighed separately, described in detail in later paragraphs.
If the option of accepting the job is chosen, the schedule is updated by tight-scheduling the newly released job in the interval $[d-p,d]$ where $p$ and $d$ are the processing time and deadline of the newly released job, respectively. Then the part of the previous schedule after time $d-p$ is moved to start at time $d$, or the end of the current schedule, whichever comes later in time (see the red job in Fig. \ref{fig:9} for illustration). This moving may lead to some of the moved jobs to miss their deadlines. Therefore the schedule is again updated to remove the part of the moved jobs that comes after their deadlines.
The decision process can be interpreted as follows. When the scheduler decides to accept the newly released job, the job is profitable once accepted but difficult to accommodate into the current schedule. Therefore in order to accommodate the newly released profitable job, the scheduler sacrifices the jobs in the current schedule in the time interval $[d-p,d]$, some of which may have deadlines far into the future, thus still have potential in completion even after the moving.

To give the procedure to compute the profit associated with the two options, we first define the notions of \textit{peace-scheduled} and \textit{contention-scheduled} jobs. We term a job \textit{peace-scheduled} if it is scheduled without
affecting other already scheduled jobs (by the appendable statement on line \ref{line:B1}), and \textit{contention-scheduled} if it is scheduled
with moving some already admitted and unfinished jobs to a later time (by the not appendable statement on line \ref{line:B2} to \ref{line:B4}).

In between two consecutive admission decisions of contention-scheduled jobs, all the jobs are peace-scheduled and the accepted jobs are always appended in the end of the schedule. The procedure to determine the profit for accepting and declining the jobs that cannot be appended can be described as follows.
First execute (virtually) on the current tentative schedule the procedure associated with the decision to accept the difficult-to-accommodate job (including scheduling the newly released job in $[d-p,d]$ and postponing the previous jobs in $[d-p,d]$) and find out the jobs in the current tentative schedule that are affected in the received processing time. Denote by $\mathcal{J}_{\tiny\mbox{affect}}$ the set of jobs in the current tentative schedule that are affected in the received processing time.
The profit associated with the option of declining can be computed as the value of the subset of jobs in $\mathcal{J}_{\tiny\mbox{affect}}$ anticipated to complete by the current tentative schedule, less the portion of penalty attributed to the subset of jobs.
The profit associated with the option of accepting can be computed as the value of the newly released job, less the portion of penalty attributed to the acceptance of the newly released job (due to affecting the jobs in $\mathcal{J}_{\tiny\mbox{affect}}$). See Fig. \ref{fig:5}, Fig. \ref{fig:51}, Fig. \ref{fig:9} and Fig. \ref{fig:10} for illustration of the profits and the decision after comparing the profits. Fig. \ref{fig:11} depicts the final schedule of this instance.

To summarize, the dynamics of the system with DSC scheduler can be
described as follows: the scheduler maintains a tentative schedule at all times; when a job request is released, the scheduler checks
whether it is possible to append the new job at the end of the current tentative schedule
while meeting its deadline. If the deadline can be met,
then the job is admitted and appended in the end of the current tentative schedule. Otherwise, the scheduler determines
whether to admit the job based on the profits of the options of accepting and declining. If the profit associated with accepting is not sufficiently large, then the job is simply declined service. Otherwise, the job is scheduled in the time interval $[d_i-p_i,d_i]$; the previous schedule after time $d_i-p_i$ is then moved to start at time $d_i$, or the end of the current schedule, whichever comes later in time, and the scheduler further checks whether there are any moved jobs that already missed their deadlines after the moving, deletes them and moves the jobs accordingly to fill the gap left by the jobs deleted.

We now describe the details of the DSC algorithm.
The pseudo code of DSC is given below. At time $0$ the
scheduler starts the infinite loop in which the schedule is updated upon each job release.

\noindent\rule{0.48\textwidth}{0.2pt} \textbf{DSC Scheduling Algorithm} procedure

\begin{algorithmic}[1]

\LOOP

\STATE upon event: job $T_{arr}$ is released

\IF{$T_{arr}$ appendable}

\STATE append $T_{arr}$ to the end of the tentative schedule;
\label{line:B1}

\ELSE

\IF{$\mbox{Profit}_{accept}>(1+\beta)\mbox{Profit}_{decline}$}\label{line:B2}

\STATE append $T_{arr}$ at the end by $d_{arr}$\label{line:B3}

\STATE move and modify the schedule after $d_{arr}-p_{arr}$ accordingly \label{line:B31}

\ELSE

\STATE decline $T_{arr}$\label{line:B4}

\ENDIF

\ENDIF

\ENDLOOP
\end{algorithmic}
\noindent\rule{0.48\textwidth}{0.2pt}

As indicated in the algorithm pseudo code $T_{arr}$ gets admitted and appended to the current schedule if it is appendable (line
\ref{line:B1}). Otherwise, the profits $\mbox{Profit}_{accept}$ and $\mbox{Profit}_{decline}$
associated with admitting and declining $T_{arr}$ respectively get compared. If admitting $T_{arr}$ assumes better profit (line
\ref{line:B2}), then $T_{arr}$
is admitted and appended at the end by $d_{arr}$ (\textit{i.e.}, scheduled in the time interval $[d_{arr}-p_{arr},d_{arr}]$),
and the current schedule after $d_{arr}-p_{arr}$ is moved and modified accordingly (line
\ref{line:B3} and \ref{line:B31}). Otherwise, if admitting $T_{arr}$ does not have
better profit, $T_{arr}$ is declined service (line
\ref{line:B4}).
The threshold $1+\beta$ will be optimized after we derive the
competitive ratio as a function of $\beta$ (see Section \ref{sec:proof}). The appendable case takes $O(1)$ per job.
while the non-appendable case takes $O(n)$ per job, where $n$ is the number of jobs in the current schedule.

We state the differences between the DSC algorithm and
the algorithms in \cite{Baruah:RTS92,Koren&Shasha:SIAMJC95} for the situation without commitment: the contention of the processor is resolved using the profit (\textit{i.e.}, job values minus penalties) instead
of the job value alone.

\begin{figure}\centering\begin{psfrags}\psfrag{t}[c]{$t$}
  \includegraphics[width=1.8in]{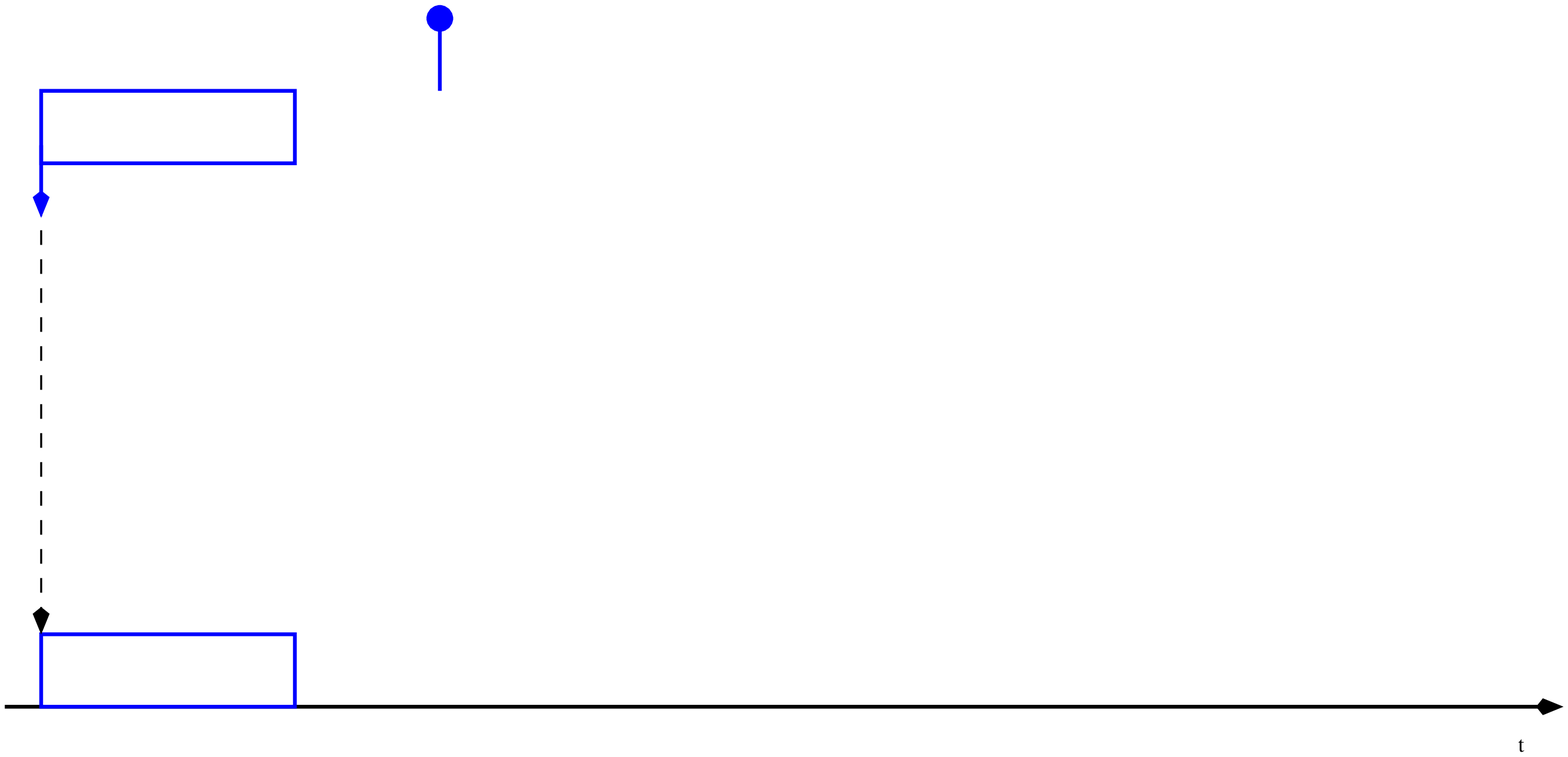}
  \caption{}\label{fig:0}
\end{psfrags}\end{figure}
\begin{figure}\centering\begin{psfrags}\psfrag{t}[c]{$t$}
  \includegraphics[width=1.8in]{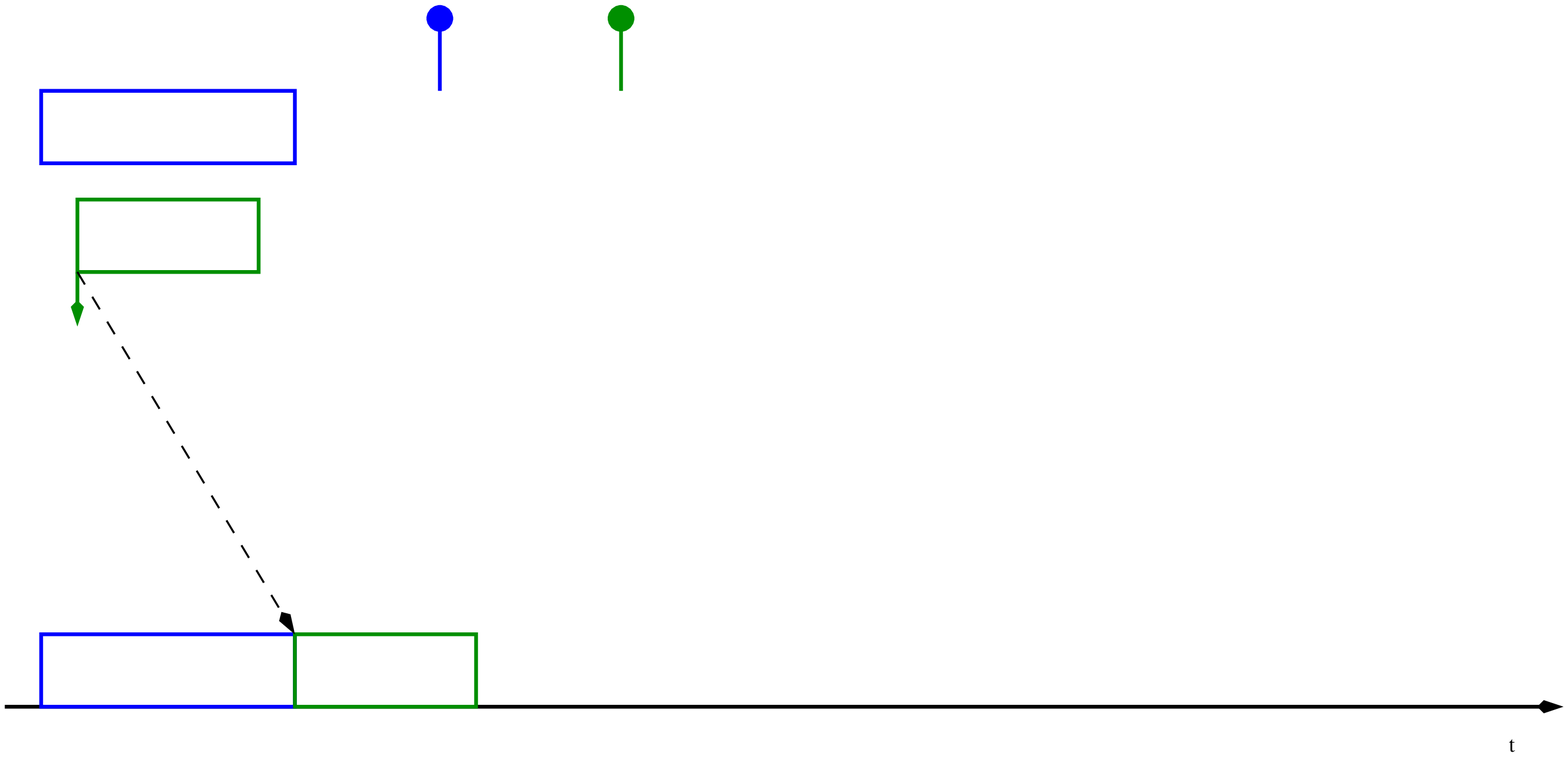}
  \caption{}\label{fig:1}
\end{psfrags}\end{figure}
\begin{figure}\centering\begin{psfrags}\psfrag{t}[c]{$t$}
  \includegraphics[width=1.8in]{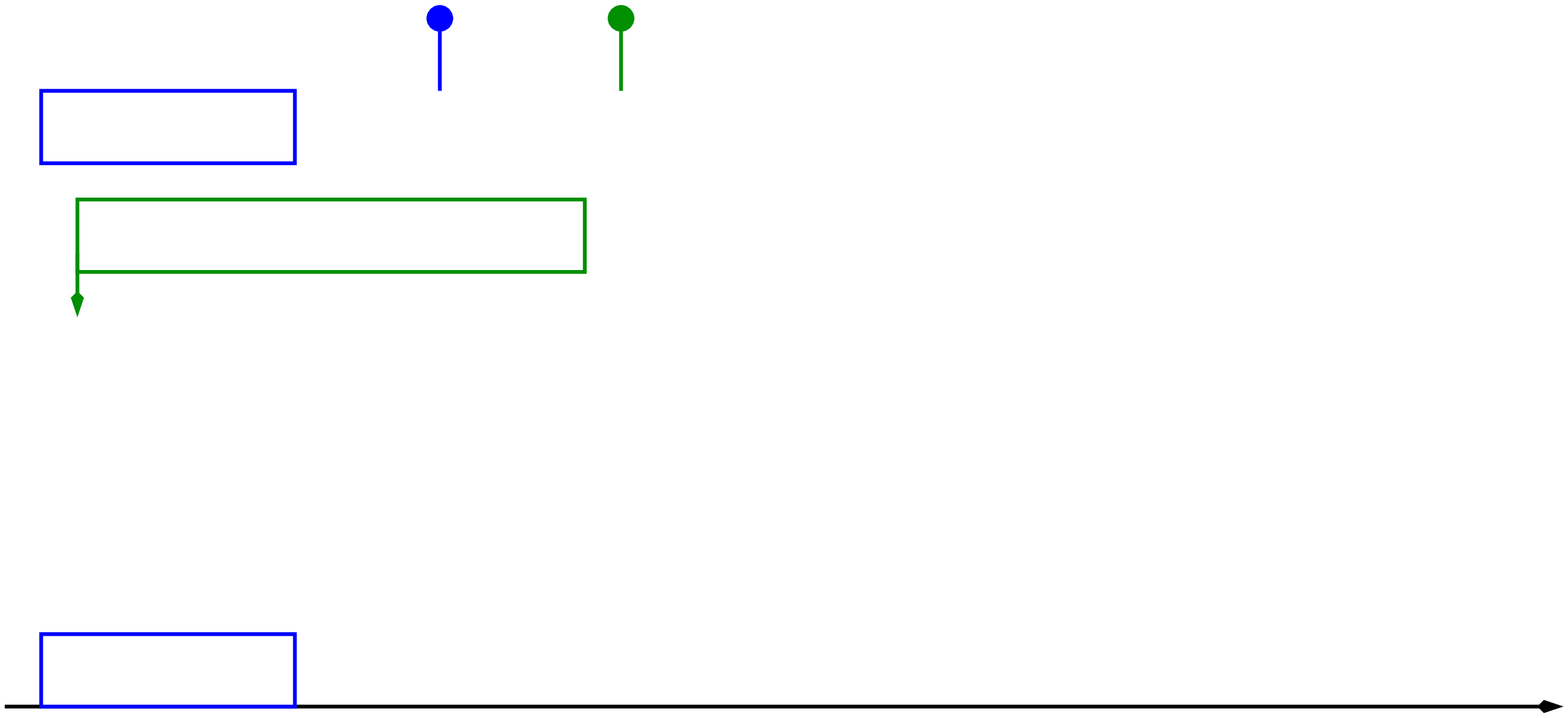}
  \caption{}\label{fig:3}
\end{psfrags}\end{figure}
\begin{figure}\centering\begin{psfrags}\psfrag{t}[c]{$t$}
  \includegraphics[width=1.8in]{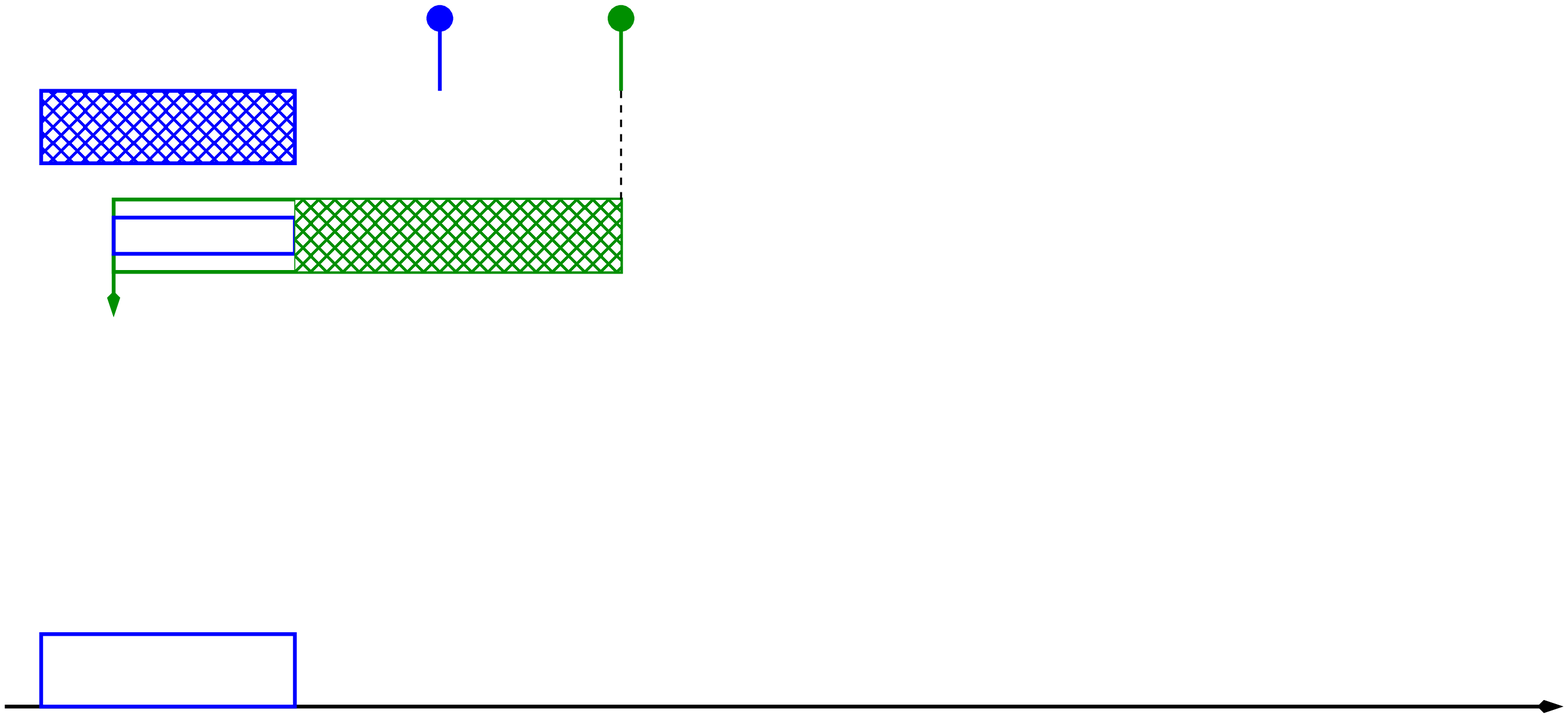}
  \caption{}\label{fig:5}
\end{psfrags}\end{figure}
\begin{figure}\centering\begin{psfrags}\psfrag{t}[c]{$t$}
  \includegraphics[width=1.8in]{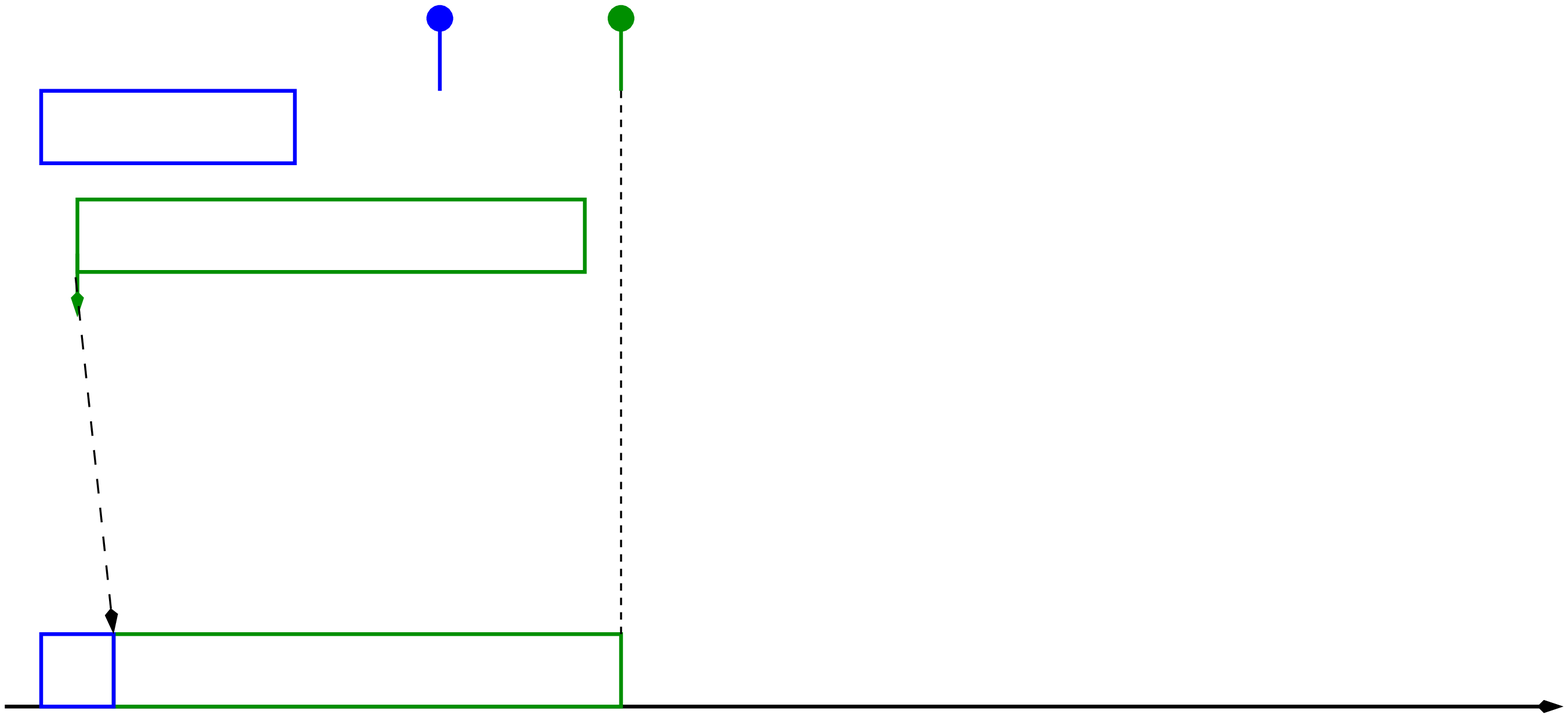}
  \caption{}\label{fig:51}
\end{psfrags}\end{figure}
\begin{figure}\centering\begin{psfrags}\psfrag{t}[c]{$t$}
  \includegraphics[width=1.8in]{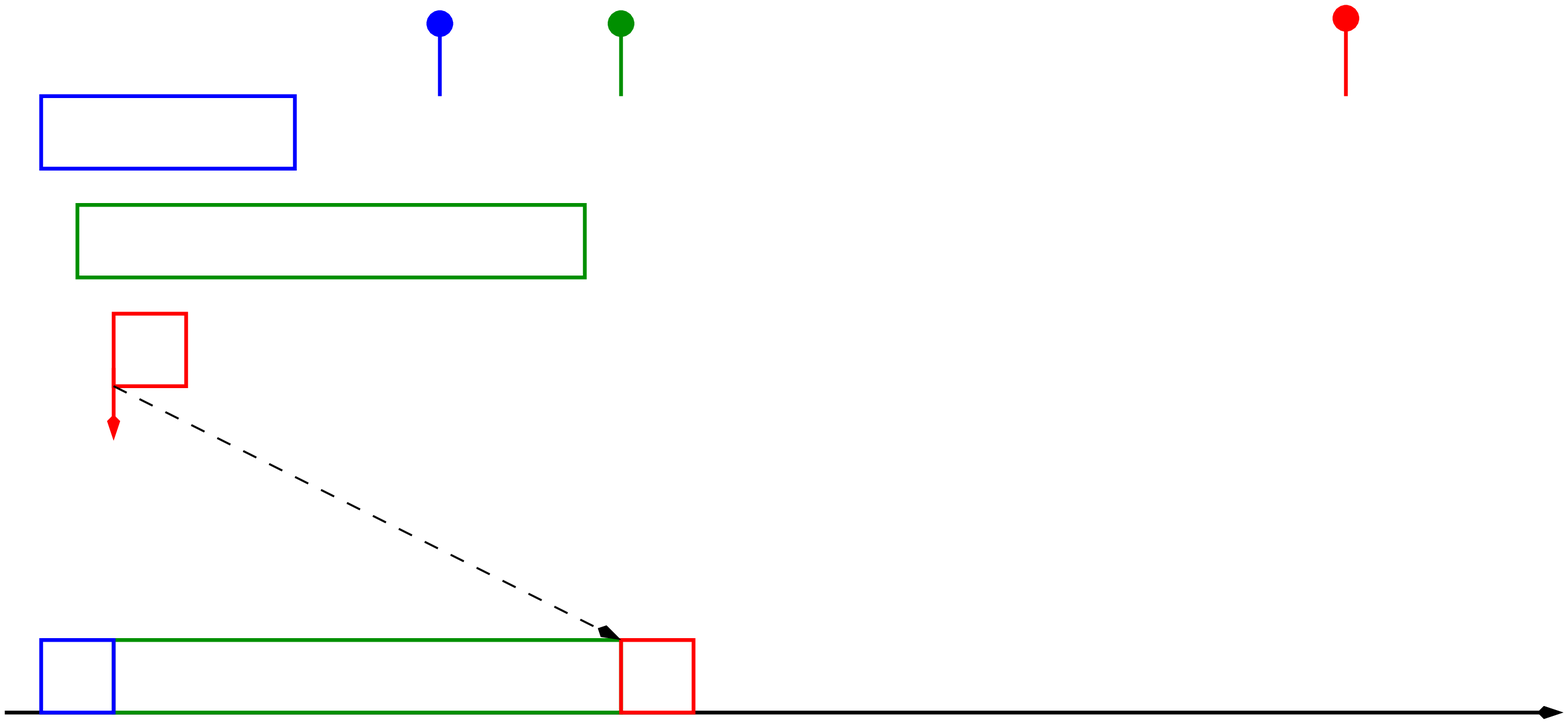}
  \caption{}\label{fig:52}
\end{psfrags}\end{figure}
\begin{figure}\centering\begin{psfrags}\psfrag{t}[c]{$t$}
  \includegraphics[width=1.8in]{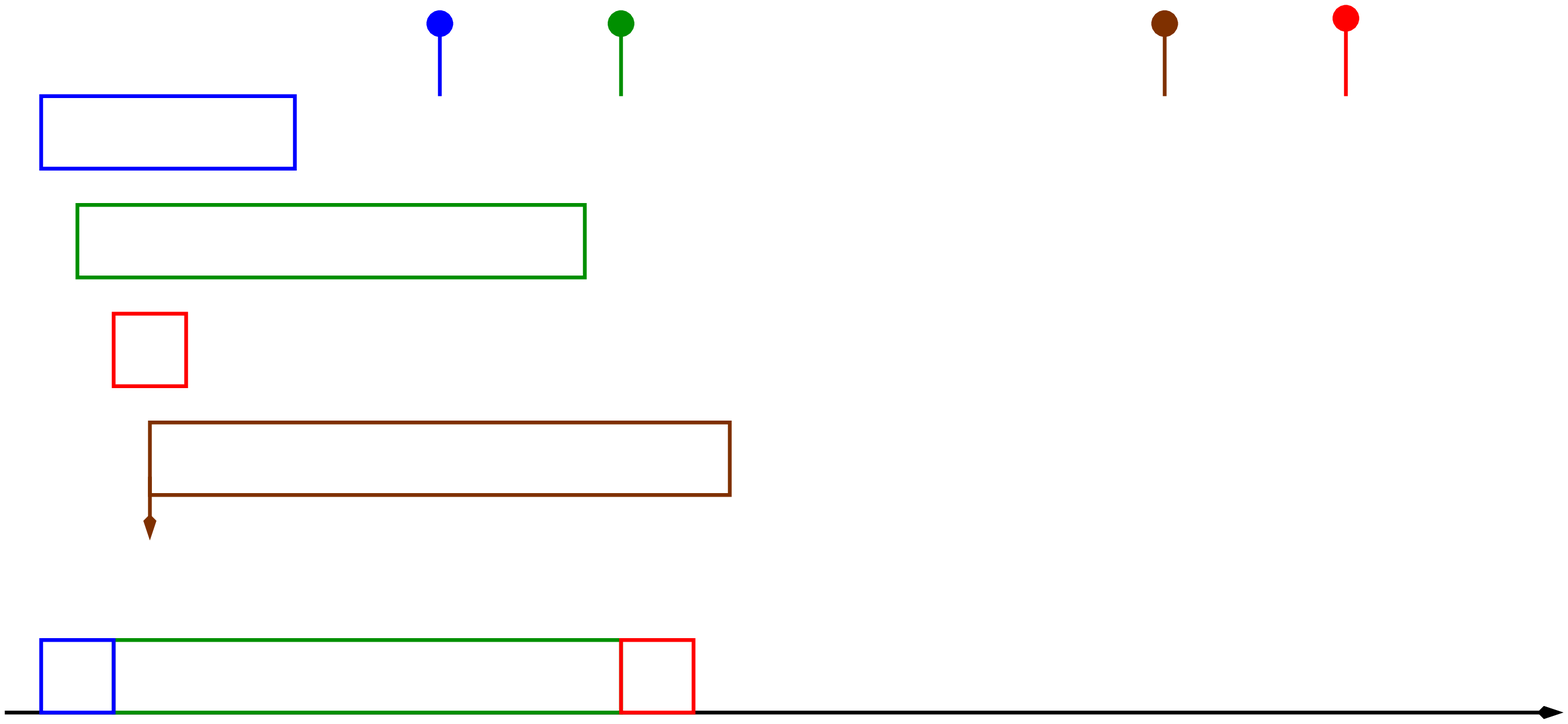}
  \caption{}\label{fig:6}
\end{psfrags}\end{figure}
\begin{figure}\centering\begin{psfrags}\psfrag{t}[c]{$t$}
  \includegraphics[width=1.8in]{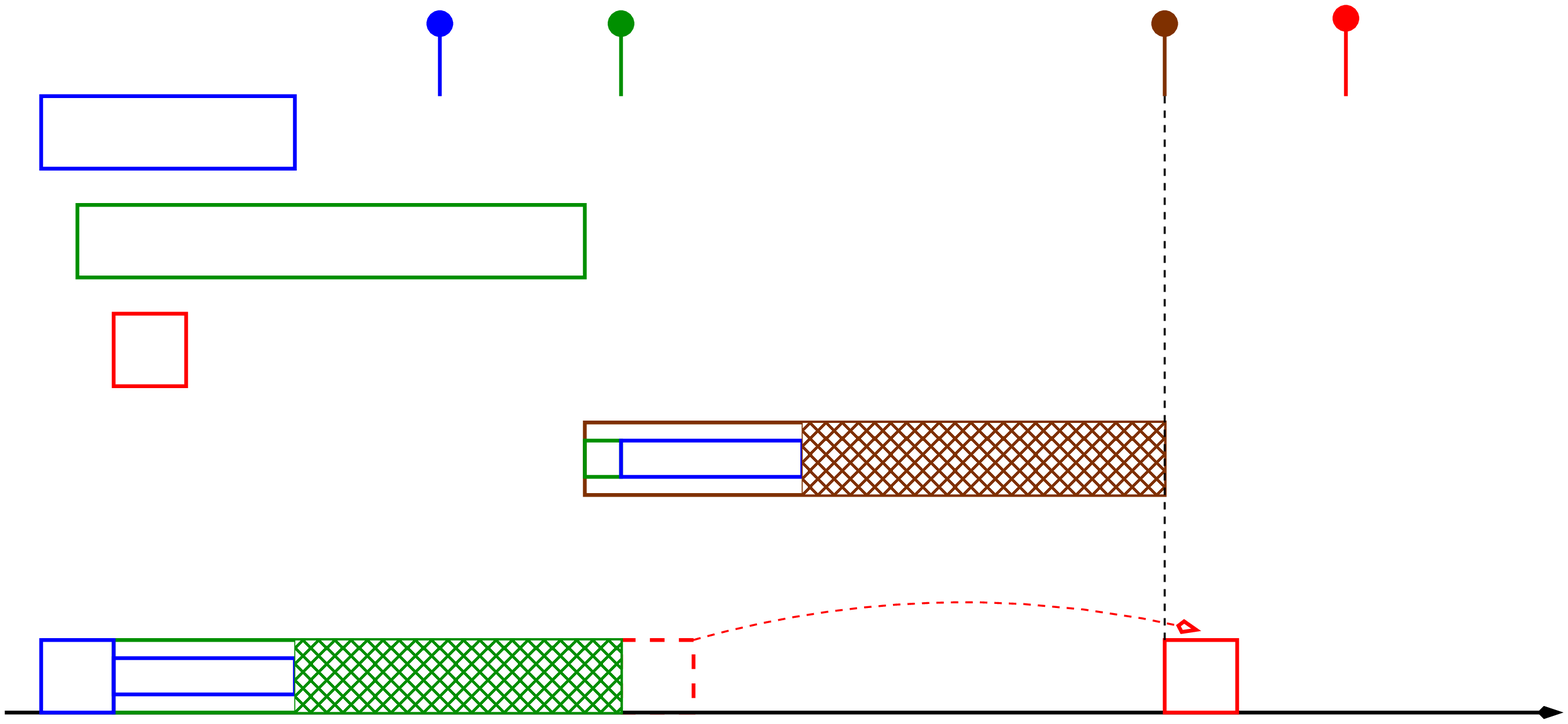}
  \caption{}\label{fig:9}
\end{psfrags}\end{figure}
\begin{figure}\centering\begin{psfrags}\psfrag{t}[c]{$t$}
  \includegraphics[width=1.8in]{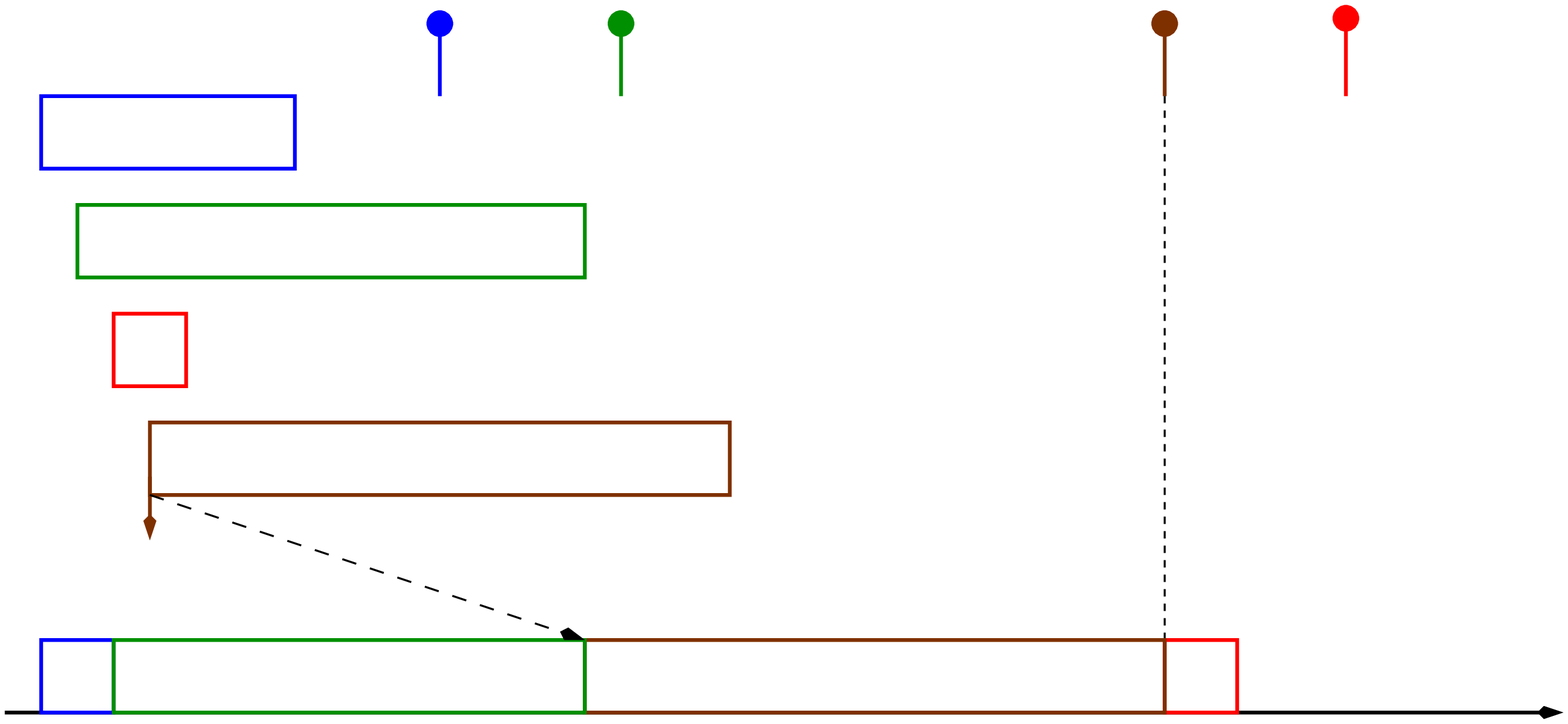}
  \caption{}\label{fig:10}
\end{psfrags}\end{figure}
\begin{figure}\centering\begin{psfrags}\psfrag{t}[c]{$t$}
  \includegraphics[width=1.8in]{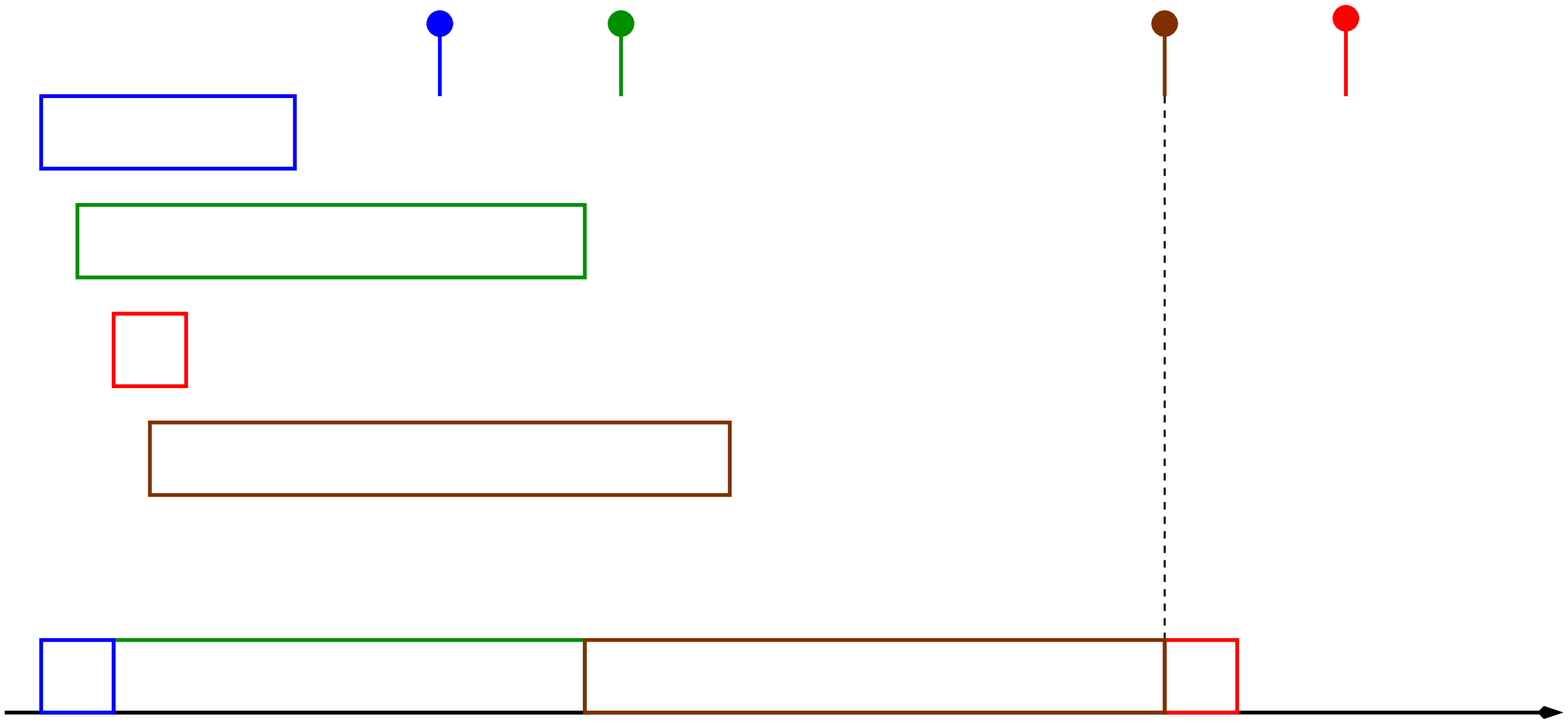}
  \caption{}\label{fig:11}
\end{psfrags}\end{figure}

\subsection{Analyzing the Structure of DSC Algorithm}

Denote a continuous busy interval (a continuous time interval in which the processor is busy executing jobs) by $B=[\underline{t},\overline{t}]$. We start the analysis of the structure of DSC algorithm by classifying the continuous busy intervals created by the execution of DSC into two different types with
different structures. The first type busy interval corresponds
to the situation where there is no processing time corresponding to contention-scheduled jobs.
In this case, all jobs admitted are
peace-scheduled and finish successfully by $\overline{t}$, the time at which the processor finishes the tentative schedule and gets idle.
The second type busy interval corresponds to the situation
where there are some contention-scheduled jobs inside the continuous busy interval.

Denote by $\textbf{\mbox{T}}_B$, $\textbf{\mbox{P}}_B$
and \textbf{C}$_B$ the total profit obtained in schedule from all jobs,
from peace-scheduled jobs and from contention-scheduled
jobs during $B$, respectively.
Note that the penalty is included in the profit $\textbf{\mbox{T}}_B$ and \textbf{C}$_B$. For
every continuous busy interval $B$ it holds that
\begin{equation}\label{eqn:sum}
\textbf{\mbox{T}}_B=\textbf{\mbox{P}}_B
+\textbf{\mbox{C}}_B.
\end{equation}

Denote by $\mathscr{B}$ the union of all continuous busy
intervals. The length of $\mathscr{B}$ will be denoted by $|\mathscr{B}|$. We
refer to its total, peace-scheduled and contention-scheduled value obtained in schedule
by \textbf{T}, \textbf{P} and \textbf{C}, respectively.

Lemma \ref{lemma:capacity} upper bounds the total processing time
in a continuous busy interval $B$ with \textbf{T}$_B$,
\textbf{P}$_B$ and \textbf{C}$_B$.

\begin{lemma}\label{lemma:capacity}
The total processing time of $B=[\underline{t},\overline{t}]$
satisfies
\begin{equation}
|B|\leq\textbf{\mbox{P}}_B+(1+\frac{1}{\beta-1})\textbf{\mbox{C}}_B=\textbf{\mbox{T}}_B+\frac{1}{\beta-1}\textbf{\mbox{C}}_B.
\end{equation}
\end{lemma}

Lemma \ref{lemma:nolargevalue} upper bounds the deadlines of the jobs
that are declined during the continuous busy interval $B$. Note that there are no jobs declined
when the processor is idle under DSC algorithm.

\begin{lemma}\label{lemma:nolargevalue}
Suppose $T_i$ was declined during the continuous busy interval
$B=[\underline{t},\overline{t}]$. Then
$$d_i-\overline{t}-\sum_{r_j\in B}s_j\leq(1+\beta)\textbf{\mbox{T}}_B,$$
where $s_j$ is the shortage (at time $\overline{t}$) of job $T_j$ that is admitted in $B$ ($r_j\in B$).
\end{lemma}

Lemma \ref{lemma:peace} provides a useful fact for the peace-scheduled jobs that eventually fail.
\begin{lemma}\label{lemma:peace}
Suppose $T_i$'s are peace-scheduled jobs that eventually failed. Then
$T_i$'s are such that $[r_i,d_i]\subset \mathscr{B}$.
\end{lemma}

\begin{proof}
If $T_i$ is peace-scheduled at time $r_i$, then $T_i$ is included in the schedule since time $r_i$. Assume that $[r_i,d_i]\subset \mathscr{B}$ does not hold. Therefore there exists time instant $\tau\in[r_i,d_i]$ such that the processor is idle at time $\tau$. However, this contradicts the way the DSC algorithm runs due to the following argument. At time $\tau$ $T_i$ is not finished yet because $T_i$ failed eventually. Therefore the scheduler can work on $T_i$ at time $\tau$ in the tentative schedule either with the goal of completing $T_i$ for its value or with the goal of reducing the penalty associated with $T_i$. Thus either way the processor should be busy, contradicting the assumed fact that the processor is idle at time $\tau$. This contradiction proves $[r_i,d_i]\subset \mathscr{B}$.
\end{proof}

\subsection{Upper Bounding Optimal Offline Value}

Given a collection of jobs $I$, denote the optimal value that an
offline algorithm can obtain from scheduling the set of jobs $I$
by $S^\ast_{\mbox{\tiny offline}}(I)$. We derive an upper bound of
$S^\ast_{\mbox{\tiny offline}}(I)$ for $I$ being the set of released jobs. We partition the
collection of jobs $I=S^c\cup S^p\cup F^p\cup F^c\cup D$ where $S^c$ ($S^p$) denotes the
successful contention-scheduled (peace-scheduled) jobs, $F^c$ ($F^p$) denotes the
failed contention-scheduled (peace-scheduled) jobs and
$D$ the declined jobs under DSC algorithm.

Since $S^\ast_{\mbox{\tiny offline}}(S^c\cup S^p\cup F^p\cup F^c\cup D)\leq S^\ast_{\mbox{\tiny offline}}(S^p)+S^\ast_{\mbox{\tiny offline}}(S^c\cup F^c\cup F^p\cup D)$, we upper bound
the two terms separately.
We upper bound the term $S^\ast_{\mbox{\tiny offline}}(S^c\cup F^p\cup F^c\cup D)$ by considering the optimal offline algorithm for $S^c\cup F^p\cup F^c\cup D$ under a \textit{processing-time-granted-value} setting.
(The granted-value setting is first used in \cite{Koren&Shasha:SIAMJC95}
to treat the no commitment case.) Specifically,
the offline scheduler receives an additional \textit{granted value}
besides the value obtained from $S^c\cup F^c\cup F^p\cup D$. The amount of granted value
depends on the offline schedule: unit value will be
granted for unit processing time in $\mathscr{B}$ that is not used for
executing jobs in $S^c\cup F^p\cup F^c\cup D$.

Under the granted-value setting the optimal offline scheduler must
consider that scheduling a job might reduce the granted value (since
processing time in $\mathscr{B}$ is used). Executing a job $T_i$ results in
a gain of $v_i$ and a loss of the granted value for the processing
time of $T_i$ that is executed in $\mathscr{B}$.

One offline schedule under the granted-value setting is to schedule
no jobs in $S^c\cup F^p\cup F^c\cup D$ (therefore leaving the entire $\mathscr{B}$ period
untouched) and get only the (whole) granted value. This
scheduling-nothing schedule obtains a value of
$|\mathscr{B}|$. Since Lemma
\ref{lemma:capacity} upper bounds the total processing time in a
continuous busy interval, we can upper bound the total processing time in
$\mathscr{B}$, and thus the value obtained by the scheduling-nothing
schedule.

However, the optimal offline schedule under the granted-value
setting may use some processing time of $\mathscr{B}$ to schedule certain
jobs in $S^c\cup F^p\cup F^c\cup D$ to obtain more value than the scheduling-nothing
schedule. To upper bound $S^\ast_{\mbox{\tiny offline}}(S^c\cup F^p\cup F^c\cup D)$ under the granted-value setting, we
take the value of the scheduling-nothing schedule as a benchmark and
turn to upper bounding the net gain the optimal schedule can have
over the scheduling-nothing schedule by completing some jobs in $S^c\cup F^p\cup F^c\cup D$.

We first observe that any job $T_i\in F^c\cup S^c$ will be such that $[r_i,d_i]\subset \mathscr{B}$, since at the time $r_i$, $T_i$ is contention-scheduled in the interval $[d_i-p_i,d_i]$. Therefore the busy period covers $[r_i,d_i]$.
We also observe by Lemma \ref{lemma:peace} that a peace-scheduled job $T_f\in F^p$ which eventually fails
also satisfies $[r_i,d_i]\subset \mathscr{B}$.
By the definition of the granted value we can see that under the optimal offline algorithm, no job is
scheduled entirely in $\mathscr{B}$ because the granted value lost would be
equal to the value of the job. Therefore the optimal offline schedule will not choose to schedule any jobs in $S^c\cup F^p\cup F^c$.

Since we are interested in scheduling jobs in $S^c\cup F^p\cup F^c\cup D$ such that only
small amount of $\mathscr{B}$ processing time is used (thus small loss of
granted value), we leverage the fact that when a job $T_d\in D$ is declined during busy interval $B$, the deadline of $T_d$ can not be too far with respect to the end of $B$, given by Lemma \ref{lemma:nolargevalue}. Lemma \ref{lemma:BU_SY} provides the earliest time
for an offline scheduler to execute a job in $D$ outside $\mathscr{B}$.

\begin{lemma}\label{lemma:BU_SY}
Suppose $T_d\in D$ is declined by the online scheduler at
time $r_d$ and $r_d\in
B=[\underline{t},\overline{t}]$. Then, if $T_d$ is to be executed
by the offline scheduler anywhere outside $\mathscr{B}$ it must be after
$\overline{t}$.
\end{lemma}

\begin{proof}
The proof can be easily done using the fact $r_d\in
B=[\underline{t},\overline{t}]$, leading to $[r_d,\overline{t}]\subset \mathscr{B}$.
\end{proof}

Lemma \ref{lemma:net} upper bounds the net gain the optimal offline
scheduler will obtain over the scheduling-nothing benchmark,
restricted to the jobs that are declined during $B$.
\begin{lemma}\label{lemma:net}
Under the granted-value setting the total net gain obtained by the
offline algorithm from scheduling the jobs in $S^c\cup F^p\cup F^c\cup D$ released in
$B=[\underline{t},\overline{t}]$ is no greater than
$(1+\beta)\textbf{\mbox{T}}_B+\sum_{r_j\in B}s_j$.
\end{lemma}

\begin{proof}
According to Lemma \ref{lemma:nolargevalue}
if $T_i$ was declined during the busy interval
$B=[\underline{t},\overline{t}]$. Then
$$d_i-\overline{t}\leq(1+\beta)\textbf{\mbox{T}}_B+\sum_{r_j\in B}s_j,$$
where $s_j$ is the shortage (at time $\overline{t}$) of job $T_j$ that is admitted in $B$ and the summation of $r_j\in B$ is summing over all jobs $T_j$ that are admitted in $B$.

On the other hand under the granted-value setting the net gain obtained by the
offline algorithm from scheduling the jobs in $S^c\cup F^p\cup F^c\cup D$ released in
$B=[\underline{t},\overline{t}]$ can only come from $D$ and the earliest
time $T_d$ can be executed
by the offline scheduler outside $\mathscr{B}$ is $\overline{t}$.

Therefore the net gain obtained by the
offline algorithm from scheduling the jobs in $S^c\cup F^p\cup F^c\cup D$ in
$B=[\underline{t},\overline{t}]$ is bounded by $$\max_{T_i\in D_B}d_i-\overline{t}\leq(1+\beta)\textbf{\mbox{T}}_B+\sum_{r_j\in B}s_j,$$
where $D_B$ is the subset of $D$ released in $B$.
\end{proof}

Lemma \ref{lemma:shortage} upper bounds the total shortage $\sum_{r_j\in B}s_j$ (at time $\overline{t}$).
\begin{lemma}\label{lemma:shortage}
$\sum_{r_j\in B}s_j\leq |B|-\textbf{\mbox{T}}_B$.
\end{lemma}

\begin{proof}
Since each $T_j$ that is not finished in $B$ contributes $s_j$ to $\sum_{r_j\in B}s_j$ and $-s_j$ to $\textbf{\mbox{T}}_B$, and each $T_j$ that is finished in $B$ contributes $0$ to $\sum_{r_j\in B}s_j$, $v_j$ to $\textbf{\mbox{T}}_B$ and $v_j$ to $|B|$, the lemma is proved.
\end{proof}

\subsection{Proving Theorem \ref{thm:1}}\label{sec:proof}

We now prove Theorem \ref{thm:1} after bounding the net gain of
scheduling the jobs in $S^c\cup F^p\cup F^c\cup D$.
\begin{proof}
Lemma \ref{lemma:net} bounds the maximum net gain for each busy
interval. By construction, each job is accounted for in exactly one
continuous busy interval. Therefore, summing over all busy
intervals we conclude using Lemma \ref{lemma:shortage} that under the granted value setting the total net gain during the entire
execution horizon obtained by the offline algorithm from scheduling the jobs of $S^c\cup F^p\cup F^c\cup D$ is bounded by $|\mathscr{B}|+\beta\textbf{\mbox{T}}$, where $\mathscr{B}$ is the union of all the busy intervals.

Combining the upper bound of the total net gain and the value of
the scheduling-nothing benchmark, \textit{i.e.},
the processing time in $\mathscr{B}$, yields
\begin{eqnarray}
\nonumber S^\ast_{\mbox{\tiny offline}}(I) &\leq& S^\ast_{\mbox{\tiny offline}}(S^c\cup F^p\cup F^c\cup D)+S^\ast_{\mbox{\tiny offline}}(S^p) \\
\nonumber &\leq& |\mathscr{B}|+(|\mathscr{B}|+\beta\textbf{\mbox{T}})+\mbox{value}(S^p) \\
\nonumber &\leq& 2|\mathscr{B}|+\beta\textbf{\mbox{T}}+\mbox{value}(S^p) \\
\label{eqn:all}&\leq& (\beta+2)\textbf{\mbox{T}}+\frac{2}{\beta-1}\textbf{\mbox{C}}+\textbf{\mbox{P}} \\
\label{eqn:all2}&\leq& (\beta+2)\textbf{\mbox{T}}+\frac{2}{\beta-1}\textbf{\mbox{C}}+\frac{2}{\beta-1}\textbf{\mbox{P}} \\
\nonumber &\leq& (\beta+2+\frac{2}{\beta-1})\textbf{\mbox{T}} \\
\label{eqn:all3}&\leq& (3+2\sqrt{2})\textbf{\mbox{T}},
\end{eqnarray}
where Eq. (\ref{eqn:all}) is obtained from
summing Lemma \ref{lemma:capacity} over all continuous busy intervals,
Eq. (\ref{eqn:all2}) holds when $\beta\leq3$, and Eq. (\ref{eqn:all3}) is obtained by optimizing over $\beta$, which yields $\beta=1+\sqrt{2}$ and
\begin{equation}
 S^\ast_{\mbox{\tiny offline}}(I)\leq(3+2\sqrt{2})\textbf{\mbox{T}},
\end{equation}
where $I=S^p\cup S^c\cup F^p\cup F^c\cup D$ is the set of released jobs.

Since the value of the optimal offline schedule is at most $S^\ast_{\mbox{\tiny offline}}(I)$ and the profit obtained by DSC
algorithm is $\textbf{\mbox{T}}$, the competitive ratio $3-2\sqrt{2}$ is shown to be achievable by DSC algorithm.
\end{proof}

\section{Upper Bound on Competitive Ratio}

An adversary argument establishes the upper bound on the competitive ratio ever achievable by any online scheduler.
Specifically, we construct a
job input instance $I$, such that the competitive ratio for the constructed job input instance
is upper bounded by $3-2\sqrt{2}$.

Consider the input instance $I$ constructed
by the adversary which contains a sequence of tight jobs ($r_i+p_i=d_i$). The first tight job $T_0$ is released at time $0$ with processing length $x_0=1$. The offline adversary observes the action of the online scheduler and then decides future job releases.
Upon the release of $T_0$ the online scheduler can choose either to admit or to decline the job $T_0$.

If declined, the offline adversary will choose to release no more jobs and eventually the offline adversary obtains $x_0$ while the online scheduler $0$.

If admitted, the adversary will choose to release another tight job $T_1$ at time $\epsilon$, with processing length $x_1$. Then similarly the
online scheduler can choose either to admit or to decline the job $T_1$ upon arrival.

Similarly, if declined, the offline adversary will choose to release no more jobs and eventually the offline adversary obtains $x_1$ while the online scheduler $x_0$.

If admitted, the adversary released another tight job $T_2$ at time $2\epsilon$, with processing length $x_2$.

The whole process keeps going until the online scheduler chooses to decline the first job in the process. For the $(n+2)$th release the offline adversary releases the tight job $T_{n+1}$ at time $(n+1)\epsilon$, with processing length $x_{n+1}$. Then similarly the
online scheduler can choose either to admit or to decline the job $T_{n+1}$ upon arrival.

If declined, the offline adversary will choose to release no more jobs and eventually the offline adversary obtains $x_{n+1}$ while the online scheduler $x_n-\sum_{j=0}^{n-1}{x_j}$, where $\sum_{j=0}^{n-1}{x_j}$ is the non-completion penalty paid by the online scheduler (the non-completion penalty should ideally include a term with $\epsilon$, however, the adversary will choose $\epsilon$ to be arbitrarily small and the term can be left out in the following derivation).

For the above job release up to $T_{m+1}$ (\textit{i.e.}, even the online scheduler chooses to admit up to job $T_{m+1}$, the offline adversary will not release new jobs), the competitive ratio ever achievable for the above constructed input instance is
\begin{eqnarray}\label{eqn:cr}
\max\{\sigma_1,\sigma_2,\ldots,\sigma_{n+1},\ldots,\sigma_{m+1},\frac{x_{m+1}-\sum_{j=0}^{m}{x_j}}{x_{m+1}}\},
\end{eqnarray}
where $\sigma_{i+1}=\frac{x_i-\sum_{j=0}^{i-1}{x_j}}{x_{i+1}}$, for $i=0,1,\ldots,m$.

Now we design the processing lengths $x_i$ to upper bound the value of Eq. (\ref{eqn:cr}). We first set all terms but the last inside the minimum in Eq. (\ref{eqn:cr}) to be $1/c$.
\begin{eqnarray}
\nonumber cx_0&=&x_1 \\
\nonumber c(x_1-x_0)&=&x_2 \\
\nonumber &\vdots& \\
\label{eqn:eq3} c(x_n-\sum_{j=0}^{n-1}{x_j})&=&x_{n+1}\\
\label{eqn:eq2} c(x_{n+1}-\sum_{j=0}^{n}{x_j})&=&x_{n+2}\\
\nonumber &\vdots& \\
\nonumber c(x_m-\sum_{j=0}^{m-1}{x_j})&=&x_{m+1}
\end{eqnarray}

We can then obtain the recursion by subtracting Eq. (\ref{eqn:eq3}) from Eq. (\ref{eqn:eq2}),
\begin{equation}\label{eqn:recursion1}
x_0=1,\quad x_1=c,\quad c(x_{n+1}-2x_n)=x_{n+2}-x_{n+1},
\end{equation}
with the characteristic function
\begin{equation}
x^2-(c+1)x+2c=0.
\end{equation}

We still need the last term inside the maximum in Eq. (\ref{eqn:cr}) to be no greater than $1/c$,
\begin{equation}\label{eqn:need}
\frac{x_{m+1}-\sum_{j=0}^{m}{x_j}}{x_{m+1}}\leq\frac{1}{c}.
\end{equation}

Rewrite Eq. (\ref{eqn:need}) to be
\begin{equation}
\frac{x_{m+1}-\sum_{j=0}^{m}{x_j}}{x_{m+1}}\leq\frac{x_m-\sum_{j=0}^{m-1}{x_j}}{x_{m+1}},
\end{equation}
which implies
\begin{equation}
x_{m+1}\leq 2x_m,
\end{equation}
and further (due to Eq. (\ref{eqn:recursion1}))
\begin{equation}\label{eqn:need2}
x_{m+2}\leq x_{m+1}.
\end{equation}

For any $1\leq c<3+2\sqrt{2}$ the characteristic function has two complex roots and there exists $m$ such that Eq. (\ref{eqn:need}) is satisfied. Therefore we can use $c$ arbitrarily close to $3+2\sqrt{2}$ to construct the sequence of tight jobs with processing length $x_0,x_1,\ldots,x_{m+1}$, for which the best competitive ratio ever achievable is $1/c$. Taking the limit of $c\to3+2\sqrt{2}$ yields the conclusion that the best competitive ratio ever achievable is $(3+2\sqrt{2})^{-1}=3-2\sqrt{2}$, matching the competitive ratio of DSC.

\section{Conclusion}

We consider the problem of online preemptive job scheduling
with deadlines and commitment requirement for the application of PHEV garage charging scheduling. We propose an online scheduling algorithm DSC
and analyze its competitive ratio in the presence of overload.
We show that the competitive ratio of DSC is $3-2\sqrt{2}=17.16\%$.
We also show that no online scheduling
algorithm can achieve a better competitive ratio, which establishes the optimality of DSC competitive-ratio-wise. Comparing with the optimal competitive ratio of $1/4=25\%$ without the commitment requirement, our result quantifies the performance loss (in terms of competitive ratio) due to the commitment obligation to be $7.84\%$.
The multi-processor scheduling and the average performance of the DSC algorithm under a stochastic setup will be investigated in future work.

\bibliographystyle{IEEEtran}
{\footnotesize\bibliography{Allerton'11}}

\end{document}